\documentclass[11pt,a4paper]{article}

\usepackage[T1]{fontenc}
\usepackage[utf8]{inputenc}
\usepackage{lmodern}
\usepackage[margin=1in]{geometry}
\usepackage{setspace}
\usepackage{microtype}
\usepackage{amsmath,amssymb,amsfonts,amsthm,bm,mathtools}
\usepackage{graphicx}
\usepackage{booktabs}
\usepackage{threeparttable}
\usepackage{array}
\usepackage{enumitem}
\usepackage{algorithm}
\usepackage{algpseudocode}
\usepackage{authblk}
\usepackage[dvipsnames]{xcolor}
\usepackage{hyperref}
\usepackage[nameinlink,noabbrev]{cleveref}
\usepackage[numbers]{natbib}

\hypersetup{
  colorlinks=true,
  linkcolor=MidnightBlue,
  citecolor=MidnightBlue,
  urlcolor=MidnightBlue
}

\newtheorem{theorem}{Theorem}

\newtheorem{assumption}{Assumption}
\newtheorem{remark}{Remark}

\DeclareMathOperator{\tr}{tr}

\DeclareMathOperator*{\argmin}{arg\,min}

\newcommand{\Prob}{\mathbb{P}}
\newcommand{\I}{\mathbf{I}}
\newcommand{\one}{\mathbf{1}}
\newcommand{\calN}{\mathcal{N}}
\newcommand{\calW}{\mathcal{W}}
\newcommand{\KL}{\mathrm{KL}}
\newcommand{\J}{\mathrm{J}}
\newcommand{\Bhat}{\mathrm{B}}

\title{\textbf{Bootstrap-Calibrated Spectral Divergence Tests\\
for Online Detection of Covariance Matrix Changes}}

\author[1]{Mehmet S\i dd\i k \c{C}ad\i rc\i}
\author[2,*]{Martin Singull}
\affil[1]{Faculty of Science, Department of Statistics,
          Cumhuriyet University, Sivas, T\"urkiye}
\affil[2]{Department of Mathematics,
          Link\"oping University, Link\"oping, Sweden}
\affil[*]{Corresponding author:
          \href{mailto:martin.singull@liu.se}{\texttt{martin.singull@liu.se}}}
\date{}

\begin{document}
\maketitle

\begin{abstract}
A covariance matrix rarely distorts in a single direction. A shift can expand all
variances simultaneously, shift the variance in one or two principal directions, displace the spectral mass without changing the marginal means, or simply
rotate the dependence structure. The mean-shift detectors that most practitioners use by default are blind to all these effects. At the same
time, we are not aware of any online procedure that calibrates a family of spectral deviation tests with provable false alarm control across both time and tracking window selection; this paper aims to fill that gap. Four spectral
deviations—$D_{\KL}(P_{1}\|P_{0})$, $D_{\KL}(P_{0}\|P_{1})$, Jeffreys, and
Bhattacharyya; each evaluated on the eigenvalues of the empirical relative
covariance operator: $\widehat{\Sigma}_{0}^{-1/2}\widehat{\Sigma}_{t}\widehat{\Sigma}_{0}^{-1/2}$.
The critical values were obtained from a conditional parametric bootstrap method that simultaneously accounts for
the estimation uncertainty in the covariance estimates in both the past and tracking windows;
this is an aspect that asymptotic approaches typically
overlook. The resulting procedure controls the false alarm rate on a family-by-family basis within a predefined monitoring period and a range of candidate window sizes;
when a single operational window is required, a power-based criterion
selects it. Theoretically, we prove the constant alternative consistency under standard conditions
by deriving local spectral expansions that exhibit second-order sensitivity near the null hypothesis. Our simulation results are conservative under the null hypothesis
and, more interestingly, show that the detection power depends largely on the spectral \emph{shape}
rather than just the magnitude of the change:
While $D_{\KL}(P_{1}\|P_{0})$ outperforms under global inflation, Jeffreys and
Bhattacharyya covers a broader range of alternatives. To illustrate the practical importance of
joint covariance tracking, we provide examples using three financial
applications: European stock indices, Fama-French sector portfolios,
and a group of large-cap technology stocks.

\end{abstract}


\section{Introduction}

Point-of-change analysis, at its core,
aims to identify structural breaks in a time series before they occur or as soon as they occur. The origin of the sequential formulation
is based on \citet{page1955}, and since then, this field has been
extensively studied by \citet{basseville1993}, \citet{aminikhanghahi2017}, and \citet{truong2020}. A
large portion of this literature, understandably, focuses on the mean. However, in multivariate
environments, a break in the mean is generally the least interesting scenario.
Financial markets, industrial sensors, biological measurements, and network traffic can exhibit sudden changes in volatility, cross-sectional dependence, or
covariance geometry, even while the mean vector remains virtually stationary. Detecting such
changes requires tools based specifically on second-order structure, rather than
relying on mean-shift mechanisms.

Part of the challenge is that “covariance change” is not a single entity.
Global variance inflation, a sudden increase limited to one or two main
directions, a redistribution of spectral mass, and
a complete reversal of the dependence structure—all of these are valid covariance changes, and each
tends to leave a different signature on any test statistic. A
detector tuned to one of these scenarios may completely miss the others.
While the current literature addresses parts of this picture, it
fails to fully piece the whole together: \citet{aue2009} has derived asymptotic covariance stability
tests under standard regularity conditions for multivariate time series;
\citet{avanesov2018} has proposed a multiscale bootstrap procedure for
high-dimensional covariance structures, and \citet{li2023} has developed
online stopping rules for the same problem. Regarding mean drift,
\citet{wang2018changepoint} introduced sparse
projection estimators for high-dimensional changes, and more recently, \citet{bao2026spectral}
proposed a spectral procedure for online covariance tracking
based on linear spectral statistics of sample Fisher matrices. Offline methods, such as the PELT algorithm in
\citet{killick2012}, are effective for retrospective
segmentation but, by their very nature, cannot be run on a live stream. What none of these approaches
offer is a method that is valid simultaneously
over time and in selecting the tracking window. 

One of the natural ways to measure how far apart two covariance matrices have become
in terms of the eigenvalues of their relative covariance operator
is through information-geometric deviations. For Gaussian distributions
, the Kullback–Leibler, Jeffreys, and Bhattacharyya divergences
all reduce to explicit functions of the eigenvalues of
$\Sigma_{0}^{-1/2}\Sigma_{1}\Sigma_{0}^{-1/2}$; this makes them
affine-invariant with virtually no additional effort. The study \citet{cadirci2022entropy}
examined entropy-based divergences for testing multivariate distribution hypotheses
and demonstrated the consistency of kNN-based KL test statistics for generalized Gaussian distributions
under the maximum entropy principle. Closer
to the current framework, in \citet{erdtman2023}’s master’s thesis,
the same divergence measures were identified as promising
candidates, particularly for covariance tracking; however, two practical issues
: how to calibrate operational thresholds and
how to select the size of the monitoring window. It has become evident that both issues
are of great importance in practice.

This article addresses precisely these two questions. It formulates a consistent
family of four affine-invariant divergence statistics for online monitoring
 and then directly tackles the problem: after the reference covariance has been estimated and shrinkage applied, there is no computable closed-form expression for the joint null hypothesis distribution of these nonlinear log-determinant and trace functionals. Here is the approach we have outlined: We use the conditional parametric bootstrap method. This method draws samples from the adapted null model. In this resampling step, we do not spread the uncertainty across both the old and new windows. Instead, we keep one side fixed. We also address how to select the window length. If you want to track multiple windows over an observation period, you can use a Bonferroni-style rule. This keeps the family-wise false alarm rate under control throughout the observation period. Another option is available. If the  plan is to use a single window that will actually be active, we can use a power-based rule. This rule selects the window based on power. To support this approach, we use local spectral expansions around the null hypothesis. We then present an expression that remains valid when the alternative is fixed. In our simulation, we test four different deviations. We find that they respond differently to changes in covariance. Each exhibits a unique spectral pattern. Finally, we apply the method to three financial problems. The goal here is to demonstrate that the method works not only in theory but also in real-world applications.

The remainder of the article has been organized as follows.
Section~\ref{sec:model} defines the tracking model and the four spectral
deviation statistics. Section~\ref{sec:bootstrap}, in turn, addresses bootstrap
calibration, the false alarm distribution, the window selection criterion, and
consistency.  Sections~\ref{sec:simulation} and~\ref{sec:realdata} present, respectively,
the simulation study and real-data applications.
Section~\ref{sec:discussion}discusses limitations and extensions, and
the Section~\ref{sec:conclusion} presents the results.

\section{Model and Spectral Divergence Statistics}
\label{sec:model}

\subsection{Monitoring model}

Let $\{X_{t}\}_{t\geq 1}$ be a sequence of independent $p$-dimensional
random vectors.  A historical reference sample of length $n_{0}$ is observed
from an in-control regime.  The monitoring model is
\begin{equation}
X_{t} \sim
\begin{cases}
  \calN_{p}(\mu,\Sigma_{0}), & t < \tau,\\
  \calN_{p}(\mu,\Sigma_{1}), & t \geq \tau,
\end{cases}
\label{eq:model}
\end{equation}
where $\tau$ is an unknown change time and the mean vector $\mu$ is constant
throughout.  In empirical applications the series are centred before covariance
estimation, isolating second-order structure from location.

At monitoring time $t$, let $\calW_{t}(k)=\{X_{t},\ldots,X_{t+k-1}\}$ denote
a window of size $k$.  Write $\widehat{\Sigma}_{0}$ for the covariance estimator
from the historical sample and $\widehat{\Sigma}_{t}(k)$ for the covariance
estimator from $\calW_{t}(k)$.  The empirical relative covariance matrix is
\begin{equation}
\widehat{R}_{t}(k)
= \widehat{\Sigma}_{0}^{-1/2}\,\widehat{\Sigma}_{t}(k)\,\widehat{\Sigma}_{0}^{-1/2}.
\label{eq:Rhat}
\end{equation}
The window-level testing problem is
\begin{equation}
H_{0}: \Sigma_{1}=\Sigma_{0}
\qquad\text{against}\qquad
H_{1}: \Sigma_{1}\neq\Sigma_{0},
\label{eq:hypothesis}
\end{equation}
equivalently $H_{0}:\Psi=\I_{p}$ against $H_{1}:\Psi\neq\I_{p}$, where
\begin{equation}
\Psi = \Sigma_{0}^{-1/2}\Sigma_{1}\Sigma_{0}^{-1/2}
\label{eq:relative-pop}
\end{equation}
is the population relative covariance operator.

\subsection{Spectral divergences for Gaussian distributions}

Let $P_{0}=\calN_{p}(0,\Sigma_{0})$ and $P_{1}=\calN_{p}(0,\Sigma_{1})$, and
let $\lambda_{1},\ldots,\lambda_{p}$ denote the eigenvalues of $\Psi$.  Under
the Gaussian assumption, four divergence measures admit closed spectral forms.
The Kullback--Leibler divergence from $P_{1}$ to $P_{0}$ is
\begin{equation}
D_{\KL}(P_{1}\|P_{0})
=
\tfrac{1}{2}\!\left\{\tr(\Psi)-\log|\Psi|-p\right\}
=
\tfrac{1}{2}\sum_{j=1}^{p}\!\left(\lambda_{j}-\log\lambda_{j}-1\right).
\label{eq:KL10}
\end{equation}
The opposite direction is
\begin{equation}
D_{\KL}(P_{0}\|P_{1})
=
\tfrac{1}{2}\!\left\{\tr(\Psi^{-1})+\log|\Psi|-p\right\}
=
\tfrac{1}{2}\sum_{j=1}^{p}\!\left(\lambda_{j}^{-1}+\log\lambda_{j}-1\right).
\label{eq:KL01}
\end{equation}
The Jeffreys divergence, the symmetrised sum, is
\begin{equation}
D_{\J}(\Psi)
= D_{\KL}(P_{1}\|P_{0}) + D_{\KL}(P_{0}\|P_{1})
= \tfrac{1}{2}\sum_{j=1}^{p}\!\left(\lambda_{j}+\lambda_{j}^{-1}-2\right).
\label{eq:Jeffreys}
\end{equation}
The Bhattacharyya distance is
\begin{equation}
D_{\Bhat}(\Psi)
=
\tfrac{1}{2}\log\!\left(\frac{|(\I_{p}+\Psi)/2|}{|\Psi|^{1/2}}\right)
=
\tfrac{1}{2}\sum_{j=1}^{p}\log\!\left(\frac{1+\lambda_{j}}{2\sqrt{\lambda_{j}}}\right).
\label{eq:Bhatt}
\end{equation}
All four quantities are non-negative and equal zero if and only if $\Psi=\I_{p}$,
i.e., $\Sigma_{0}=\Sigma_{1}$.  Because they depend only on the generalised
eigenvalues of the pair $(\Sigma_{0},\Sigma_{1})$, they are invariant under
nonsingular affine transformations of the data.

\begin{remark}
Two-way Kullback–Leibler divergences weight the expansion
and asymmetric contraction of the covariance matrix. For $\lambda_{j}=\gamma>1$ (global inflation),
$D_{\KL}(P_{1}\|P_{0})$ exceeds $D_{\KL}(P_{0}\|P_{1})$; since
$(\gamma-\log\gamma-1)>(\gamma^{-1}+\log\gamma-1)$.
This asymmetry has direct effects on the tracking power
(Section~\ref{sec:simulation}).
\end{remark}

\subsection{Sample statistics}

Let $\widehat{\lambda}_{1,t},\ldots,\widehat{\lambda}_{p,t}$ be the eigenvalues
of $\widehat{R}_{t}(k)$.  The four window-level monitoring statistics are
\begin{align}
T_{\KL,t}^{1\to 0}(k) &= k\,D_{\KL}(P_{1}\|P_{0})\!\left(\widehat{R}_{t}(k)\right),
\label{eq:Tklr}\\
T_{\KL,t}^{0\to 1}(k) &= k\,D_{\KL}(P_{0}\|P_{1})\!\left(\widehat{R}_{t}(k)\right),
\label{eq:Tklf}\\
T_{\J,t}(k) &= k\,D_{\J}\!\left(\widehat{R}_{t}(k)\right),
\label{eq:Tj}\\
T_{\Bhat,t}(k) &= k\,D_{\Bhat}\!\left(\widehat{R}_{t}(k)\right).
\label{eq:Tb}
\end{align}
The factor $k$ stabilises the null scale in a fixed dimension and ensures that
the statistics diverge under any fixed alternative as the window size grows.

When $p$ is not negligible relative to $n_{0}$ or $k$, the sample covariance
matrices may be ill-conditioned.  We allow a shrinkage estimator of the form
\begin{equation}
\widehat{\Sigma}^{(\rho)}
= (1-\rho)\,\widehat{\Sigma}
  + \rho\,\frac{\tr(\widehat{\Sigma})}{p}\,\I_{p},
\qquad 0\leq\rho<1,
\label{eq:shrinkage}
\end{equation}
in the spirit of \citet{ledoitwolf2004}.  The same estimation rule is applied
consistently to the historical and monitoring-window estimators, and in every
bootstrap replication.

\subsection{Local spectral expansions} \label{sec:local}

The local behavior of the statistics in the region near the null hypothesis is obtained from the Taylor
expansion. Suppose $\lambda_{j}=1+h_{j}/\sqrt{k}$, where $h_{j}=O(1)$. Then,
as $k\to\infty$,
\begin{align}
k\,D_{\KL}(P_{1}\|P_{0})(\Psi) &= \tfrac{1}{4}\textstyle\sum_{j=1}^{p} h_{j}^{2}+O(k^{-1/2}),
\label{eq:loc-klr}\\
k\,D_{\KL}(P_{0}\|P_{1})(\Psi) &= \tfrac{1}{4}\textstyle\sum_{j=1}^{p} h_{j}^{2}+O(k^{-1/2}),
\label{eq:loc-klf}\\
k\,D_{\J}(\Psi) &= \tfrac{1}{2}\textstyle\sum_{j=1}^{p} h_{j}^{2}+O(k^{-1/2}),
\label{eq:loc-j}\\
k\,D_{\Bhat}(\Psi) &= \tfrac{1}{16}\textstyle\sum_{j=1}^{p} h_{j}^{2}+O(k^{-1/2}).
\label{eq:loc-b}
\end{align}
Of the four statistics, all are locally sensitive to the same quadratic spectral
deviation from $\I_{p}$, but differ in terms of their scale constants.  Specifically,
the Jeffreys statistic is four times as large as each unidirectional KL statistic
in the local regime, and the Bhattacharyya statistic is four times smaller
than the Jeffreys statistic.  These scaling differences
imply that raw statistical values corresponding to different types of deviation
cannot be compared without calibration; the bootstrap threshold is an integral part of the procedure. The derivation
of this is given in Appendix $\ref{app:proofs}$.

\section{Bootstrap Calibration and Online Monitoring}
\label{sec:bootstrap}

\subsection{Conditional parametric bootstrap}

Under $H_{0}$, both $\widehat{\Sigma}_{0}$ and $\widehat{\Sigma}_{t}(k)$
estimate the same matrix $\Sigma_{0}$.  The exact null distribution of the
statistics in \eqref{eq:Tklr}--\eqref{eq:Tb} is analytically unwieldy for
finite $n_{0}$ and $k$, particularly when shrinkage is applied and $\Sigma_{0}$
is unknown.  We therefore calibrate critical values by a conditional parametric
bootstrap.

Fix a divergence $D$ and window size $k$, and write $T_{t}(D;k)=kD(\widehat{R}_{t}(k))$.
Given the historical forecast $\widehat{\Sigma}_{0}$,
for $b = 1, \ldots, B$,
\[
X_{1,b}^{*(0)}, \ldots, X_{n_{0},b}^{*(0)}
\stackrel{\mathrm{i.i.d.}}{\sim}
\calN_{p}(0,\widehat{\Sigma}_{0}),
\qquad
Y_{1,b}^{*},\ldots,Y_{k,b}^{*}
\stackrel{\mathrm{i.i.d.}}{\sim}
\calN_{p}(0,\widehat{\Sigma}_{0}).
\]
From these compute the resampled estimators $\widehat{\Sigma}_{0,b}^{*}$,
$\widehat{\Sigma}_{t,b}^{*}$, and
\[
\widehat{R}_{b}^{*}(k)
= \left(\widehat{\Sigma}_{0,b}^{*}\right)^{-1/2}
  \widehat{\Sigma}_{t,b}^{*}
  \left(\widehat{\Sigma}_{0,b}^{*}\right)^{-1/2}.
\]
The statistics of the bootstrap, $T_{b}^{*}(D;k)=kD(\widehat{R}_{b}^{*}(k))$,
reflects the joint estimation uncertainty of both the
reference and the tracking window estimators. For a local significance level of $\alpha_{k}$,
the bootstrap critical value is as follows:
\begin{equation}
c_{k,\alpha_{k}}^{*}(D)
= \inf\!\left\{x:\frac{1}{B}\sum_{b=1}^{B}\one\bigl\{T_{b}^{*}(D;k)\leq x\bigr\}
  \geq 1-\alpha_{k}\right\},
\label{eq:boot-crit}
\end{equation}
and the alert rule is
\begin{equation}
\delta_{t}(D;k) = \one\!\left\{T_{t}(D;k)>c_{k,\alpha_{k}}^{*}(D)\right\}.
\label{eq:alert}
\end{equation}

\subsection{Family-wise false-alarm control}

Let $\mathcal{K}=\{k_{1},\ldots,k_{m}\}$ be a set of candidate window sizes
and $N_{\max}$ the expected pre-change monitoring horizon.  For each $k$,
set $L_{k}=\lceil N_{\max}/k\rceil$.  Let $\omega_{k}\geq 0$ with
$\sum_{k\in\mathcal{K}}\omega_{k}=1$ be scale weights, and define
\begin{equation}
\alpha_{k} = \frac{\alpha_{\mathrm{FWER}}\,\omega_{k}}{L_{k}}.
\label{eq:alpha-k}
\end{equation}

\begin{theorem}[Family-wise false-alarm control]
\label{thm:fwer}
Suppose the bootstrap critical value in \eqref{eq:boot-crit} satisfies
$\Prob_{H_{0}}(\delta_{t}(D;k)=1)\leq\alpha_{k}$ for each $k\in\mathcal{K}$.
Then
\[
\Prob_{H_{0}}\!\left(
\bigcup_{k\in\mathcal{K}}\,\bigcup_{\ell=1}^{L_{k}}\{\delta_{\ell}(D;k)=1\}
\right) \leq \alpha_{\mathrm{FWER}}.
\]
\end{theorem}

\begin{proof}
According to the union,
\[
\Prob_{H_{0}}\!\left(\bigcup_{k,\ell}\{\delta_{\ell}(D;k)=1\}\right)
\leq
\sum_{k\in\mathcal{K}}\sum_{\ell=1}^{L_{k}}
\Prob_{H_{0}}\{\delta_{\ell}(D;k)=1\}
\leq
\sum_{k\in\mathcal{K}} L_{k}\,\alpha_{k}
= \alpha_{\mathrm{FWER}}.
\qedhere
\]
\end{proof}

\subsection{Power-driven window selection}

In applications involving a single operational window, we select the value of $k$ based on a
design power criterion. Let us assume that $\Psi_{\star}$ is a scientifically interesting
reference alternative. For a candidate $k$ value, we estimate the following:
\begin{equation}
\pi_{k}(D;\Psi_{\star})
= \Prob_{\Psi_{\star}}\!\left\{T_{t}(D;k) > c_{k,\alpha_{k}}^{*}(D)\right\}
\label{eq:pi-k}
\end{equation}
Estimate this value using a Monte Carlo simulation based on $\calN_{p}(0,\widehat{\Sigma}_{0}^{1/2}\Psi_{\star}\widehat{\Sigma}_{0}^{1/2})$.
Assuming that the post-shift windows are approximately independent, the number of windows required to detect a signal with at least $\beta$ probability
is given by:
\begin{equation}
m_{k}(\beta,D;\Psi_{\star})
= \left\lceil\frac{\log(1-\beta)}{\log\{1-\pi_{k}(D;\Psi_{\star})\}}\right\rceil,
\label{eq:m-k}
\end{equation}
and the optimal window minimizes the expected detection delay:
\begin{equation}
k_{\mathrm{opt}}(D;\Psi_{\star})
= \argmin_{k\in\mathcal{K}}\; k\,m_{k}(\beta,D;\Psi_{\star}).
\label{eq:kopt}
\end{equation}
This criterion balances the power per window with the delay
caused by a longer window.

\begin{algorithm}[htbp]
\caption{Spectral Deviation Monitoring System Calibrated Using Bootstrap}
\label{alg:detector}
\begin{algorithmic}[1]
\Require Historical sample; tracking stream; window set $\mathcal{K}$;
         deviation $D$; bootstrap size $B$; false alarm budget
         $\alpha_{\mathrm{FWER}}$; tracking horizon $N_{\max}$.
\State Estimate $\widehat{\Sigma}_{0}$ from the historical sample.
\For{for each $k\in\mathcal{K}$}
    \State Compute $L_{k}=\lceil N_{\max}/k\rceil$ and the values of $\alpha_{k}$.
          
    \State Select $B$ bootstrap reference and tracking windows under $\calN_{p}(0,\widehat{\Sigma}_{0})$.
    \State Calculate the value of $c_{k,\alpha_{k}}^{*}(D)$ using \eqref{eq:boot-crit}.
\EndFor
\For{for each monitoring window $\calW_{t}(k)$}
    \State Compute the values of $\widehat{\Sigma}_{t}(k)$, $\widehat{R}_{t}(k)$,
           and $T_{t}(D;k)$.
    \If{$T_{t}(D;k)>c_{k,\alpha_{k}}^{*}(D)$ then, for some $k\in\mathcal{K}$}
        \State Trigger an alarm; record the observation window, the scale, and the statistical value.
    \EndIf
\EndFor
\end{algorithmic}
\end{algorithm}

\subsection{Consistency}

\begin{assumption}
\label{ass:consistency}
The historical and tracking samples are independent of each other. The $p$-dimensional
space is fixed, and both covariance estimators are consistent:
as $n_{0}\to\infty$, $\widehat{\Sigma}_{0}\xrightarrow{P}\Sigma_{0}$ and
as $k\to\infty$, $\widehat{\Sigma}_{t}(k)\xrightarrow{P}\Sigma_{1}$.
\end{assumption}

\begin{theorem}[Fixed-alternative consistency]
\label{thm:consistency}
Let $D\in\{D_{\KL}(P_{1}\|P_{0}),\,D_{\KL}(P_{0}\|P_{1}),\,D_{\J},\,D_{\Bhat}\}$.
Under assumption \ref{ass:consistency}, if $\Sigma_{1}\neq\Sigma_{0}$, then
\[
\frac{1}{k}\,T_{t}(D;k)\xrightarrow{P}D(\Psi)>0
\quad\text{and}\quad
T_{t}(D;k)\xrightarrow{P}\infty.
\]
Furthermore, under the zero-bootstrap
law, if $c_{k,\alpha_{k}}^{*}(D)=O_{P}(1)$, then $\Prob_{\Sigma_{1}}\{T_{t}(D;k)>c_{k,\alpha_{k}}^{*}(D)\}\to 1$.
\end{theorem}

\begin{proof}
The consistency of covariance estimators and the continuous mapping theorem
yield the result $\widehat{R}_{t}(k)\xrightarrow{P}\Psi$. Each divergence $D$
is continuous in the cone of positive-definite matrices, and since $D(\Psi) > 0$
whenever $\Psi \neq \I_{p}$—and since $D(\Psi) > 0$ when $\Psi \neq \I_{p}$—it follows that $D(\widehat{R}_{t}(k))\xrightarrow{P}D(\Psi) > 0$.
When multiplied by $k$, we obtain the divergence of $T_{t}(D;k)$. The bootstrap
critical value $c_{k,\alpha_{k}}^{*}(D)$, under the null hypothesis,
is by construction $O_{P}(1)$, so the detection probability converges to one.
\end{proof}

\section{Simulation Study}
\label{sec:simulation}

\subsection{Design}

We evaluate the simulation study to assess false alarm control under the null hypothesis and
detection power under various alternatives for covariance changes. Pre-change
observations are drawn from the set $\calN_{p}(\mathbf{0},\I_{p})$.  A representative
setting uses $p=10$ dimensions, a historical sample size of $n_{0}=500$, and
a tracking window of $k=50$; additional window sizes
are discussed in Appendix~\ref{app:simulation}.  Bootstrap calibration employs $B=1{,}000$ repetitions, and the Monte Carlo
power/delay estimates use $500$ and $300$ repetitions, respectively (
the repository’s \texttt{QUICK\_RUN} setting; as noted in the code, performing a final run with larger
number of repetitions would further reduce the Monte Carlo noise in the reported
figures and tables).
All simulations were performed using \textsf{Python} (\texttt{NumPy}/\texttt{SciPy}
for linear algebra and Monte Carlo routines, \texttt{pandas} for
data processing, and \texttt{matplotlib} for figures); the full code and
financial return series are available in the code repository
(see the Code Availability Statement).

We select alternatives to represent qualitatively different spectral changes.
\emph{Global inflation} $\Sigma_{1}=\gamma\I_{p}$ scales all eigenvalues
by a common factor $\gamma\in\{1.25,1.50,1.75,2.00\}$.  A \emph{single-peak}
option concentrates the change in a single primary direction, while a
\emph{double-peak} option causes distortion in two directions.  A \emph{balanced}
option increases the variance in one direction while decreasing it in the other and
preserves the relationship $|\Sigma_{1}|=|\Sigma_{0}|$.  The \emph{intensity-rotated}
alternative redistributes the spectral mass across all directions, which
represents a change in the covariance orientation.  Together, these four scenarios
encompass the qualitative diversity in covariance shifts
that led to this study.

\subsection{Results}

\paragraph{Power comparison.}
Figure~\ref{fig:simulation-power-heatmap} summarizes the power distribution in the first window
across all alternatives, and two patterns immediately stand out. Firstly,
$T_{\KL}^{1\to 0}$ is the strongest statistic under global inflation;
as $\gamma$ increases, its power rises sharply,
while $T_{\KL}^{0\to 1}$ remains distinctly conservative in the same regime precisely
the asymmetry noted in Section~\ref{sec:model}.The second pattern
has to do with symmetric statistics. $T_{\J}$ and $T_{\Bhat}$ start out less aggressive than
$T_{\KL}^{1\to 0}$ under weak inflation, but when the change contains spectral redistribution
rather than uniform scaling as in balanced and strongly rotated
states they outperform the directional KL statistics. Sudden increase
scenarios are challenging for everyone: with just one or two eigenvalues
deviating from unity, each statistic loses its power
relative to more scattered changes.

\begin{figure}[htbp]
    \centering
    \includegraphics[width=0.82\textwidth]{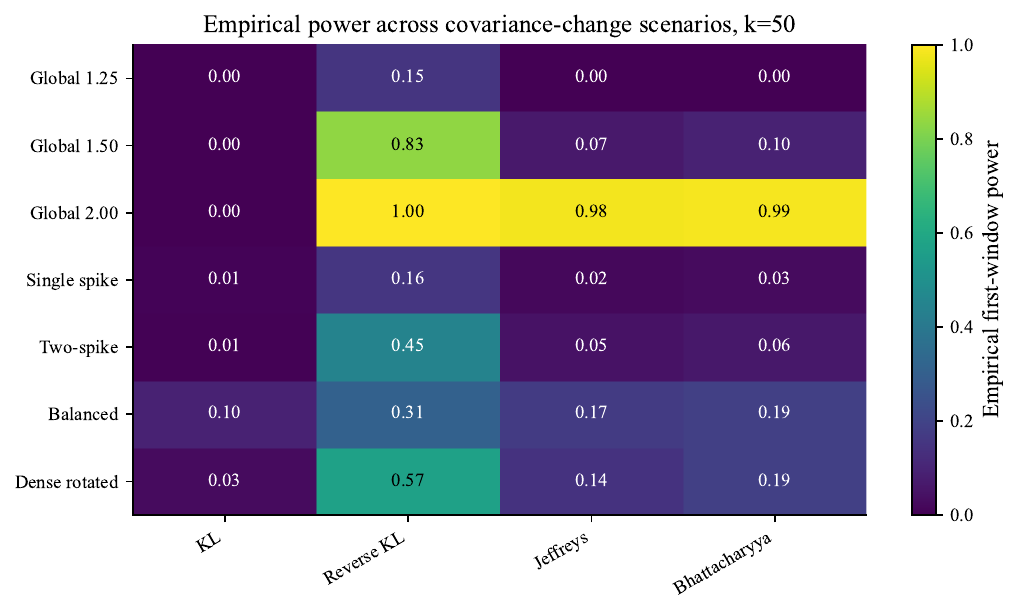}
    \caption{Among the alternatives for covariance exchange, the first-window detection power
             of four spectral deviation statistics calibrated using bootstrap
             ($p=10$, $n_{0}=500$, $k=50$).  Lighter shading
             indicates higher power.}
    \label{fig:simulation-power-heatmap}
\end{figure}

Figure~\ref{fig:simulation-power-curves} depicts how the power
varies continuously as a function of $\gamma$ under global inflation. Precisely around $\gamma=1$,
the rejection rates for all four statistics lie near the nominal level; this provides
reassuring evidence that bootstrap calibration functions
as intended near the null hypothesis. As $\gamma$ increases, $T_{\KL}^{1\to 0}$ stands out from the others and
exceeds the traditional 80
at the lowest inflation level. Jeffreys and Bhattacharyya rise more gradually yet monotonically and
eventually reach nearly perfect power under stronger inflation,
while $T_{\KL}^{0\to 1}$ never truly catches up in this scenario.

\begin{figure}[htbp]
    \centering
    \includegraphics[width=0.78\textwidth]{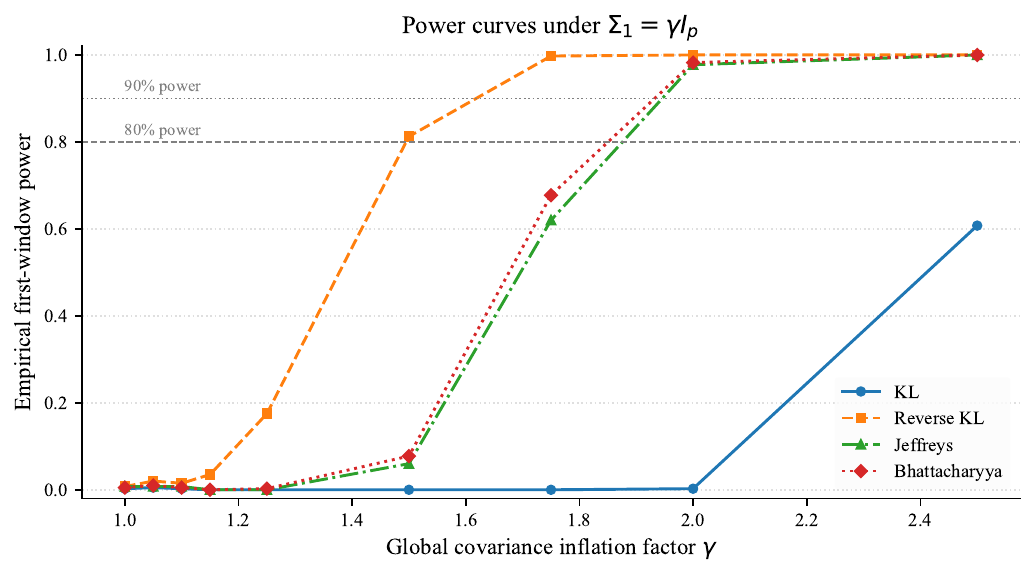}
    \caption{Power curves as a function of $\gamma$ under global covariance inflation $\Sigma_{1}=\gamma\I_{p}$.  The $T_{\KL}^{1\to 0}$ statistic
             reaches high power at the smallest inflation, while $T_{\KL}^{0\to 1}$ is
             significantly less sensitive in this configuration.}
    \label{fig:simulation-power-curves}
\end{figure}

\paragraph{Sequential detection path.}
Figure~\ref{fig:simulation-detection-path} demonstrates a typical
tracking path in the event of a global inflation change. Each statistic is plotted on the graph as a ratio
to its own bootstrap critical value; therefore, the horizontal dashed
line at one unit indicates the alarm threshold. Before the change, all four ratios
are comfortably below one; this is the state we desire so that the false alarm
control can function effectively in practice. When the change occurs, $T_{\KL}^{1\to 0}$
produces the strongest response and triggers several alarms; Jeffreys and Bhattacharyya
also exceed the threshold, but the overshoots are smaller; and $T_{\KL}^{0\to 1}$
never exceeds the threshold in this particular run—which is again
consistent with its weaker performance under inflationary conditions.

\begin{figure}[htbp]
    \centering
    \includegraphics[width=0.82\textwidth]{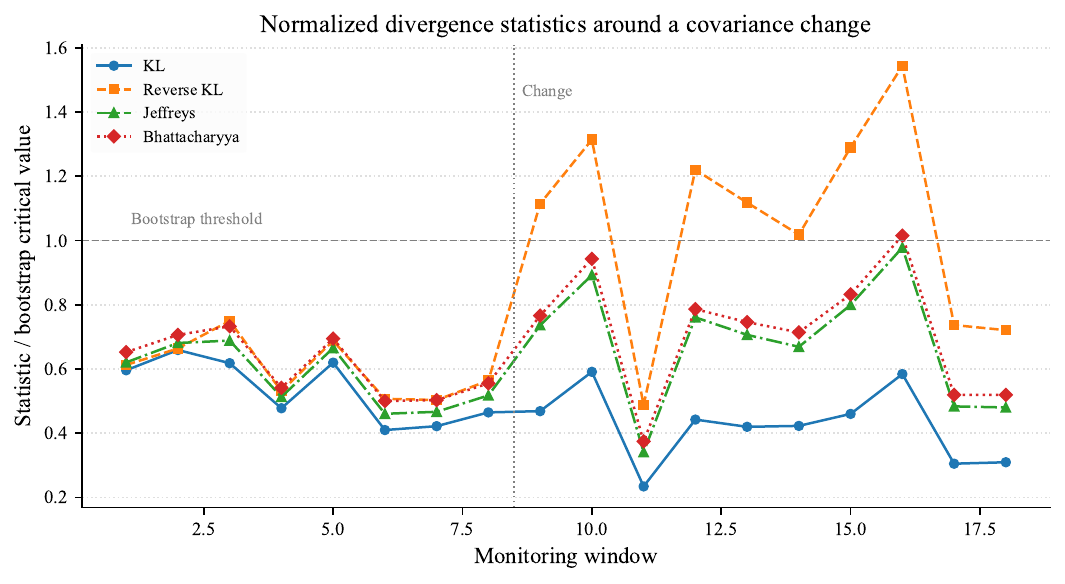}
    \caption{Normalized deviation statistics along a tracking path
             containing simulated covariance changes.  The vertical dotted line
             indicates the point of change; values above the dashed horizontal line
             correspond to an alarm condition.}
    \label{fig:simulation-detection-path}
\end{figure}

\paragraph{Numerical summary.}
Table ~\ref{tab:simulation-summary} provides numerical values for the empirical statistics and
power estimates at $k=50$. Under the null hypothesis, all three statistics shown
remain below the nominal significance level; this is the
conservative behavior we expected. The alternative hypotheses, on the other hand, present a consistent picture: global
inflation at $\gamma=1.50$ leaves essentially no
power for the $T_{\KL}^{0\to 1}$ statistic, while the Jeffreys and Bhattacharyya statistics still yield significant
probabilities of detection. The single-peak alternative hypothesis is truly
challenging for every statistic; for the perturbation occurs only at one eigenvalue and
appears nowhere else. Balanced and heavily rotated cases are
where symmetric statistics shine; the highest detection probabilities and
shortest average delays are observed in these cases; for example, the Bhattacharyya
statistic achieves a detection
probability of $0.980$ and an average delay of only $210.5$ observations in a heavily rotated setting.

\begin{table}[htbp]
\centering
\caption{Summary of simulations for bootstrap-calibrated spectral deviation
         statistics with $p=10$, $n_{0}=500$, and $k=50$.  The size of the null hypothesis
         was evaluated under covariance stability.  ``KL'' denotes $D_{\KL}(P_{0}\|P_{1})$;
         the results obtained for $D_{\KL}(P_{1}\|P_{0})$ under global inflation
         are presented in Figure~\ref{fig:simulation-power-curves}.}
\label{tab:simulation-summary}
\scriptsize
\setlength{\tabcolsep}{4pt}
\begin{tabular}{llrrrr}
\toprule
Scenario & Statistic & Size & Power & Detection probability & Mean delay \\
\midrule
Null          & KL            & 0.002 & --    & --    & --      \\
Null          & Jeffreys      & 0.004 & --    & --    & --      \\
Null          & Bhattacharyya & 0.008 & --    & --    & --      \\
\midrule
Global 1.50   & KL            & --    & 0.000 & 0.007 & 400.000 \\
Global 1.50   & Jeffreys      & --    & 0.104 & 0.833 & 361.000 \\
Global 1.50   & Bhattacharyya & --    & 0.140 & 0.920 & 319.384 \\
\midrule
Single spike  & KL            & --    & 0.018 & 0.187 & 476.786 \\
Single spike  & Jeffreys      & --    & 0.036 & 0.510 & 458.170 \\
Single spike  & Bhattacharyya & --    & 0.042 & 0.563 & 439.941 \\
\midrule
Balanced      & KL            & --    & 0.116 & 0.853 & 341.992 \\
Balanced      & Jeffreys      & --    & 0.198 & 0.977 & 233.788 \\
Balanced      & Bhattacharyya & --    & 0.212 & 0.987 & 213.514 \\
\midrule
Dense rotated & KL            & --    & 0.038 & 0.670 & 444.776 \\
Dense rotated & Jeffreys      & --    & 0.176 & 0.973 & 266.096 \\
Dense rotated & Bhattacharyya & --    & 0.244 & 0.980 & 210.544 \\
\bottomrule
\end{tabular}
\end{table}

When the simulation results are evaluated as a whole, three important conclusions emerge. Bootstrap
calibration performs well in finite samples for every window size
we examined. Detection power is determined more by the spectral \emph{shape} than by
the raw magnitude of the covariance change. Furthermore, no single statistic
is superior in all cases: $D_{\KL}(P_{1}\|P_{0})$ is the appropriate tool
when a change on a global scale is expected, but when the nature of the change is unknown, Jeffreys and Bhattacharyya are safer
default options.

\section{Real-Data Applications}
\label{sec:realdata}

\subsection{Datasets and transformation}

Three financial datasets, selected to
cover different market structures and sizes, were analyzed. The \texttt{EuStockMarkets} dataset
contains the daily closing prices of the DAX, SMI, CAC, and FTSE indices
for the period 1991–1998.  The Fama–French 10-Sector Portfolio, obtained from the Kenneth R. French Data Library,
represents a broad sectoral distribution across
non-durable consumer goods, durable consumer goods, manufacturing, energy, technology, telecommunications,
retail, healthcare, utilities, and other sectors. The \texttt{gafa\_stock} dataset from \texttt{tsibbledata}
contains the daily adjusted closing prices of Apple,
Amazon, Facebook/Meta, and Google/Alphabet from 2014 to 2018.

We convert all price series into continuous compound daily percentage
returns:
\[
r_{t} = 100\log\!\left(\frac{P_{t}}{P_{t-1}}\right).
\]
The returns have then been mean-centered prior to covariance estimation. These three datasets
capture qualitatively distinct dependency structures: cross-country stock
market co-movement (EuStockMarkets), cross-sector factor exposure and rotation
(Fama–French), and firm-level dependence on the technology sector (GAFA).

\subsection{Exploratory evidence}

Before moving on to the figure, it is worth examining
what these three series actually do; for this is what allows us to treat them
as covariance problems rather than mean problems. European index returns
(Figure~\ref{fig:EuStockMarkets_returns}) hover around zero but exhibit distinct
clusters of volatility; within these clusters, extreme movements affecting all four indices simultaneously
are observed—a sign that intermarket dependence does not
remain constant over time. The Fama-French cumulative curves
(Figure~\ref{fig:FF10_cumulative_Q1}) diverge and converge again at different points in the sample,
and this pattern of expansion and contraction
can be more naturally explained by changes in cross-sector covariance
rather than by deviations from the average in any single sector. GAFA cumulative returns
(Figure~\ref{fig:GAFA_cumulative_returns_Q1}) illustrate the same underlying story
from a different angle: several sharp common declines
are observed alongside company-specific deviations; this indicates not merely noise around a fixed target,
but a structure of common interdependence that actually changes over time.

\begin{figure}[htbp]
    \centering
    \includegraphics[width=\textwidth]{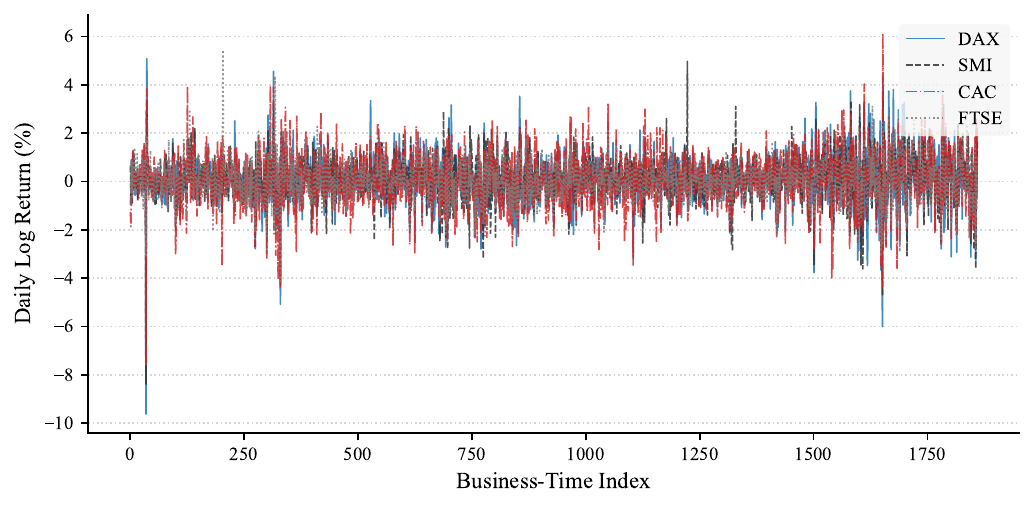}
    \caption{Daily return data for the DAX, SMI, CAC, and FTSE indices.
             The simultaneous increases in volatility across all four series
             require a joint covariance analysis
             rather than separate univariate analyses.}
    \label{fig:EuStockMarkets_returns}
\end{figure}

\begin{figure}[htbp]
    \centering
    \includegraphics[width=\textwidth]{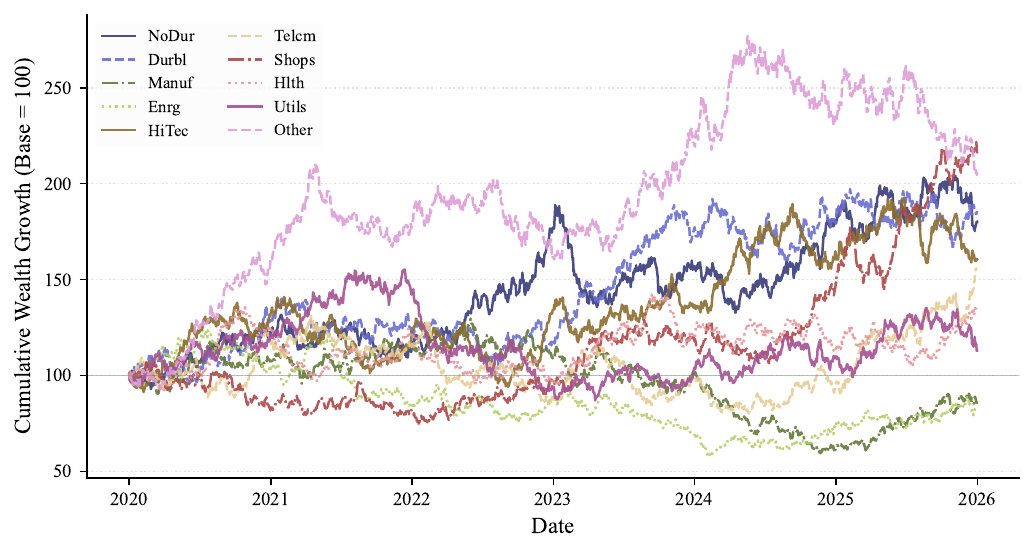}
    \caption{Fama-French 10 Sectors
             Cumulative wealth growth of the portfolios; indexed to 100 at the beginning of the observation period.
             The expanding sectoral distribution over time indicates that the covariance across sectors
             is not constant.}
    \label{fig:FF10_cumulative_Q1}
\end{figure}

\begin{figure}[htbp]
    \centering
    \includegraphics[width=\textwidth]{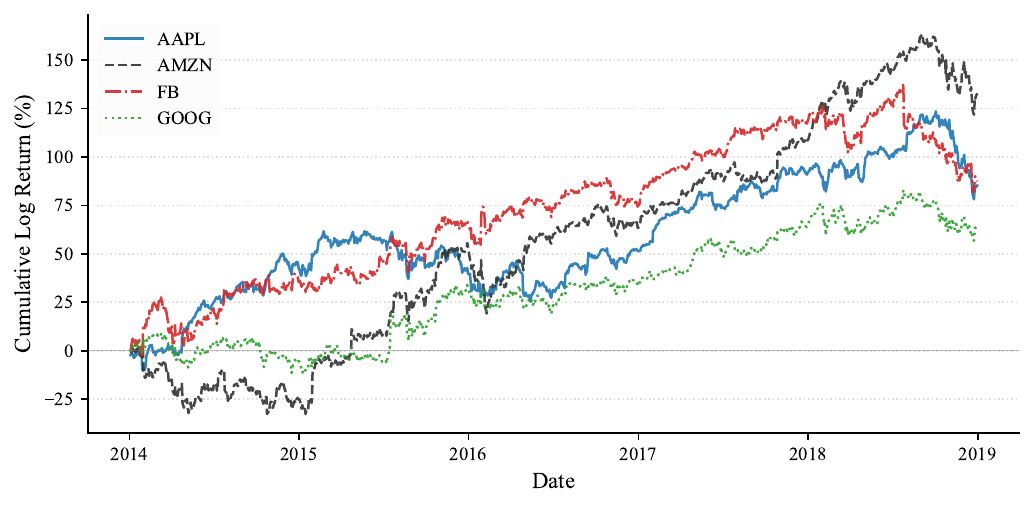}
    \caption{From 2014 to 2018, the cumulative logarithmic returns of AAPL, AMZN, FB, and GOOG stocks.  Common declines and diverging company-specific trends point to time-dependent interdependence with the technology sector.}
    \label{fig:GAFA_cumulative_returns_Q1}
\end{figure}

\subsection{Monitoring results}

We consider the first portion of the sample for each dataset as the historical
reference period and evaluate all subsequent data in $k$-sized
sliding windows. When the covariance is constant, the eigenvalues of $\widehat{R}_{t}$
cluster around one; when the eigenvalues begin to deviate from one, this reflects a change in scale,
correlation, or the direction of the principal components, and this occurs
regardless of whether the means have changed.

Figure~\ref{fig:empirical-monitoring-eustockmarkets} illustrates the monitoring
path for \texttt{EuStockMarkets}, employing the first 250 observations
as the reference period. Each deviation statistic has been plotted on the graph
as a ratio to its own bootstrap critical value; thus, anything above the horizontal line at a given value
is considered a detected change in covariance. There are several notable periods of
distinctly high values that generally coincide with the
increasing cross-market volatility experienced by European stock
markets in the late 1990s. We place greater emphasis on a
period in which several statistics collectively exceed the threshold
rather than on a single statistic exceeding the threshold on its own; for the latter may simply reflect the spectral sensitivity of the deviation in question
rather than a genuine structural break. Table~\ref{tab:real-template} summarizes these three approaches side by side.

\begin{figure}[htbp]
    \centering
    \includegraphics[width=0.82\textwidth]{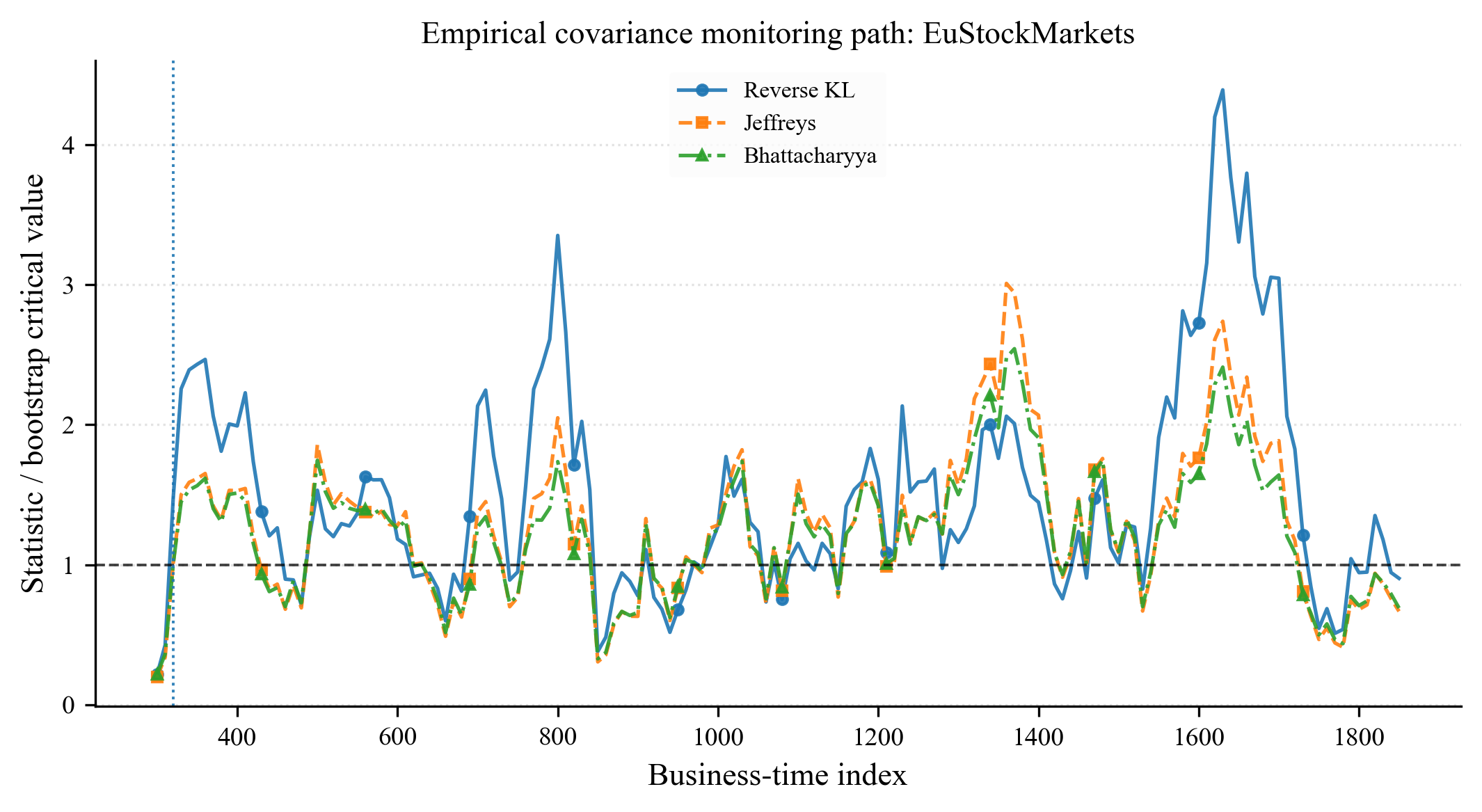}
    \caption{Covariance monitoring chart for the EuStockMarkets return series.
             The statistics have been normalized based on bootstrap critical values;
             values above the dashed line at the unit level
             indicate that a deviation from the historical covariance structure has been detected.}
    \label{fig:empirical-monitoring-eustockmarkets}
\end{figure}

\begin{table}[htbp]
\centering
\caption{Three real-data applications summarized. Each dataset
         is transformed into a multivariate return series; the reference period
         is used to estimate $\widehat{\Sigma}_{0}$, and subsequent windows
         are monitored for changes in covariance.}
\label{tab:real-template}
\small
\setlength{\tabcolsep}{3.5pt}
\begin{tabular}{p{3.0cm}p{3.4cm}p{3.0cm}p{4.6cm}}
\toprule
Dataset & Reference period & Statistics & Empirical role \\
\midrule
\texttt{EuStockMarkets}
  (DAX, SMI, CAC, FTSE)
& First 250 trading days of the 1991--1998 sample
& KL / Jeffreys / Bhattacharyya
& Low-dimensional benchmark; monitoring joint volatility and
  cross-market dependence of major European indices. \\
\midrule
Fama--French 10 Industry Portfolios
& Initial segment of the daily return sample
& KL / Jeffreys / Bhattacharyya
& Moderate-dimensional sectoral application;
  detecting sector-rotation and cross-industry covariance shifts. \\
\midrule
\texttt{gafa\_stock}
  (AAPL, AMZN, FB, GOOG)
& First year of the 2014--2018 sample
& KL / Jeffreys / Bhattacharyya
& Firm-level application; assessing time-varying dependence
  within a closely related technology-stock group. \\
\bottomrule
\end{tabular}
\end{table}

\section{Discussion}
\label{sec:discussion}

If there is one conclusion to be drawn from the simulation study, it would be that no
single deviation statistic dominates under every possible change in covariance.
If prior information points to a global variance expansion, $D_{\KL}(P_{1}\|P_{0})$
is a clear choice; however, it loses much of its advantage in the face of
contractions and spectral redistribution. Jeffreys and Bhattacharyya,
in exchange for covering a broader range of alternatives, sacrifice some power
in the case of pure scale expansion, and in practice this trade-off is generally
worth it: monitoring a small, calibrated family of statistics
is more robust than relying on a single statistic. Furthermore,
as shown in Figure~\ref{fig:empirical-monitoring-eustockmarkets},
there is a diagnostic benefit to reporting all four outcomes together which statistics
were triggered and which were not indicates not only whether a change occurred,
but also the \emph{type} of change that occurred.

It is important to identify the points where theory diverges from the data clearly. The basic assumptions—independent Gaussian observations with a constant mean—do not hold exactly for financial returns, which exhibit
serial correlation, heavy tails, and clustering of structural volatility. Therefore, the bootstrap method should be understood as operating under a calibrated Gaussian
approximation rather than a full data generation process. Properly addressing
serial correlation—perhaps through a stationary block bootstrap or
a version of the relative covariance estimator corrected by HAC—
appears to be the natural next step. In high-dimensional regimes where $p/n_{0}$ is not small,
robust covariance estimators or residuals corrected using a factor model will
likely to be helpful as well. There is also a path toward
completely relaxing the Gaussian working model: as described in \citet{cadirci2026kl}
nonparametric kNN estimator of the Kullback–Leibler divergence in \citet{cadirci2026kl}
shows that a kNN-calibrated threshold can replace the parametric bootstrap and
extend the current framework to the generalized Gaussian family
\citep{cadirci2025tsallis}.

Finally, a caveat from an empirical perspective: Analysis using real data
shows that changes in covariance can be detected in standard financial time series; however, this analysis
cannot link any identified period to a specific
macroeconomic event or policy change. For those wishing to take this study further,
combining the detection method with backwards-looking
localization techniques and providing a realistic economic interpretation would be a natural
next step.

\section{Conclusion}
\label{sec:conclusion}

This article aims to transform spectral deviation statistics into
functional tests for online covariance monitoring, and this framework essentially
achieves this: finite-sample critical values that account for estimation uncertainty in the
that account for estimation uncertainty in the covariance, a Bonferroni allocation that provides
family-wise false alarm control over time and scale, a power-focused rule for selecting
a single window, and theoretical support via local spectral expansions and fixed alternative
consistency. Thus, it fills the gap left by
\citet{erdtman2023}, which identified these deviations as promising but
did not establish a calibration mechanism around them. Simulation results
confirm that bootstrap calibration
remains conservative and stable under the null hypothesis; it is not only the magnitude but also the spectral \emph{shape} of a change
that determines which deviation statistic is most powerful.

The framework is modular enough to accommodate alternative covariance
estimators, dependency-adjusted resampling schemes, or additional spectral
distances without disrupting the underlying calibration structure—which
leaves ample room for further work on online covariance tracking
in multivariate systems.

\section*{Code Availability}

All code used to generate the simulation results, figures, and tables in this article
is publicly available at \url{https://github.com/mehmetsiddik/bootstrap-spectral-divergence}.
This repository contains a Python implementation of spectral divergence statistics calibrated using bootstrap,
the simulation code, real-data applications, and the financial return series used in Section~\ref{sec:realdata}.

\bibliographystyle{plainnat}
\bibliography{Reference}

\appendix
\section{Proof of the Local Spectral Expansions}
\label{app:proofs}
Let’s take $u$ to be a small number and expand it as usual:
\[
\log(1+u) = u - \tfrac{u^{2}}{2} + \tfrac{u^{3}}{3} + O(u^{4}),
\qquad
(1+u)^{-1} = 1 - u + u^{2} - u^{3} + O(u^{4}).
\]
If we substitute $\lambda_{j}=1+u$ into each of the four eigenvalue terms
and combine the powers of $u$, we obtain the following result:
\begin{align*}
\tfrac{1}{2}\bigl\{(1+u)-\log(1+u)-1\bigr\}
  &= \tfrac{u^{2}}{4} + O(u^{3}), \\
\tfrac{1}{2}\bigl\{(1+u)^{-1}+\log(1+u)-1\bigr\}
  &= \tfrac{u^{2}}{4} + O(u^{3}), \\
\tfrac{1}{2}\bigl\{(1+u)+(1+u)^{-1}-2\bigr\}
  &= \tfrac{u^{2}}{2} + O(u^{3}), \\
\tfrac{1}{2}\log\!\left(\frac{2+u}{2\sqrt{1+u}}\right)
  &= \tfrac{u^{2}}{16} + O(u^{3}).
\end{align*}
Now let $u = h_{j}/\sqrt{k}$ and sum over the range $j = 1, \ldots, p$. Multiplying both sides by
$k$ yields \eqref{eq:loc-klr}--\eqref{eq:loc-b}.

\section{Additional Simulation Results}
\label{app:simulation}

Figure~\ref{fig:simulation-bootstrap-calibration} graphically illustrates the empirical false alarm
probabilities for different tracking window sizes under covariance stability conditions,
compared to the nominal levels specified in \eqref{eq:alpha-k}. The curves
closely track the nominal reference value across window sizes and, in most cases,
remain above or below this value; at the largest window size,
a slight overshoot is observed for $T_{\KL}^{1\to 0}$; which is consistent with Monte Carlo noise arising from the number of repetitions used here rather than
a systematic calibration error, and
decreases with the larger numbers of repetitions recommended for the final run.

\begin{figure}[htbp]
    \centering
    \includegraphics[width=0.82\textwidth]{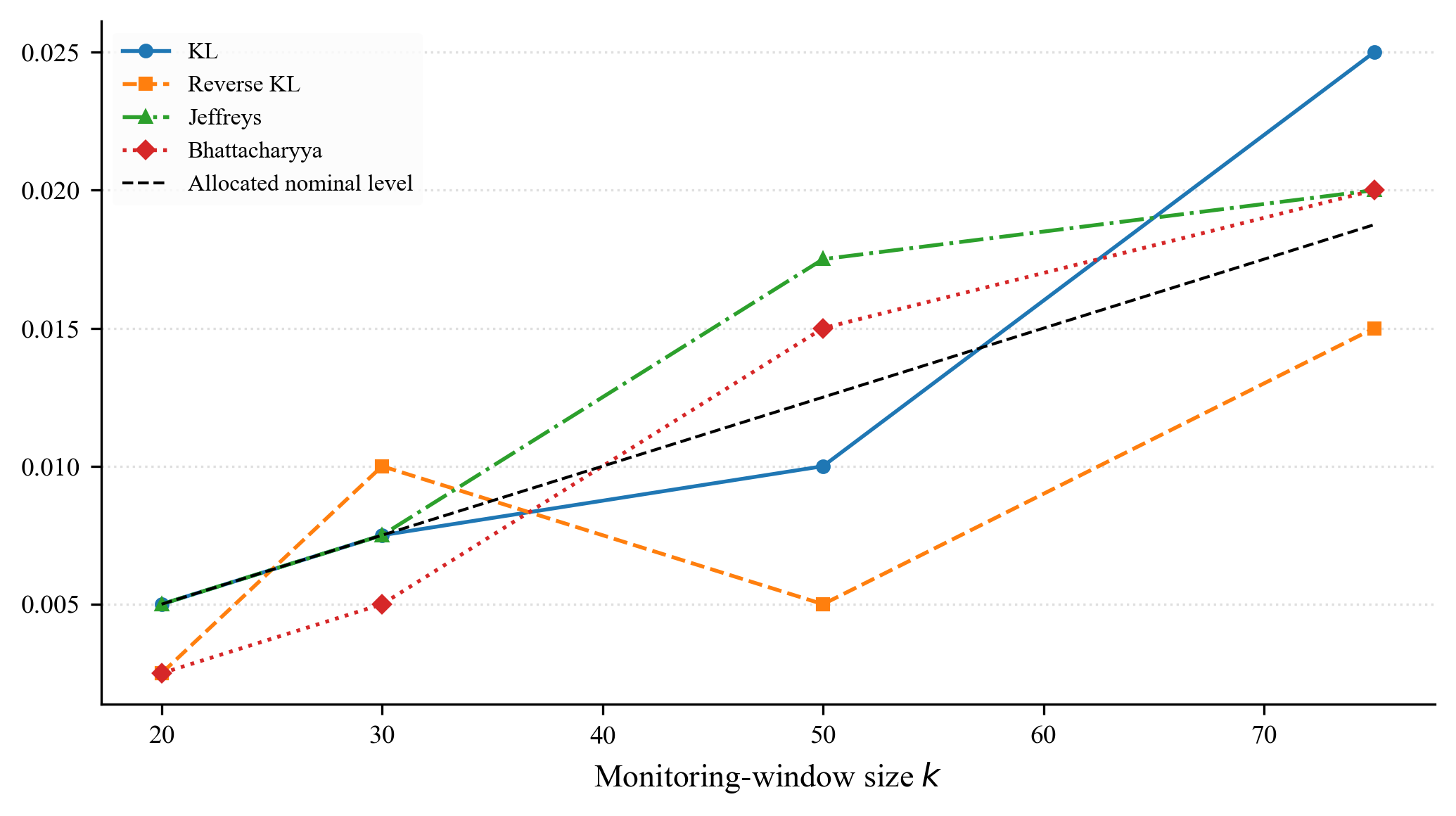}
    \caption{Under covariance stability conditions,
             empirical false alarm probabilities for window sizes $k\in\mathcal{K}$ are compared to the assigned
             nominal levels (dashed lines).  Conservative control was maintained
             throughout the entire process.}
    \label{fig:simulation-bootstrap-calibration}
\end{figure}

Four change scenarios in Figure~\ref{fig:simulation-power-delay} and
Table~\ref{tab:simulation-summary}
show the power-delay tradeoff for a representative window of $k=50$; here, the upper-left region
represents the target point: high detection probability, short delay. While Jeffreys and
Bhattacharyya fall at a more balanced point between balanced and
heavily rotated alternatives, directional KL statistics are more sensitive to the
exact form of spectral skew. The overall picture
shows that detection quality must be evaluated jointly in terms of detection probability
and delay, rather than relying solely on the power of the first window.

\begin{figure}[htbp]
    \centering
    \includegraphics[width=0.78\textwidth]{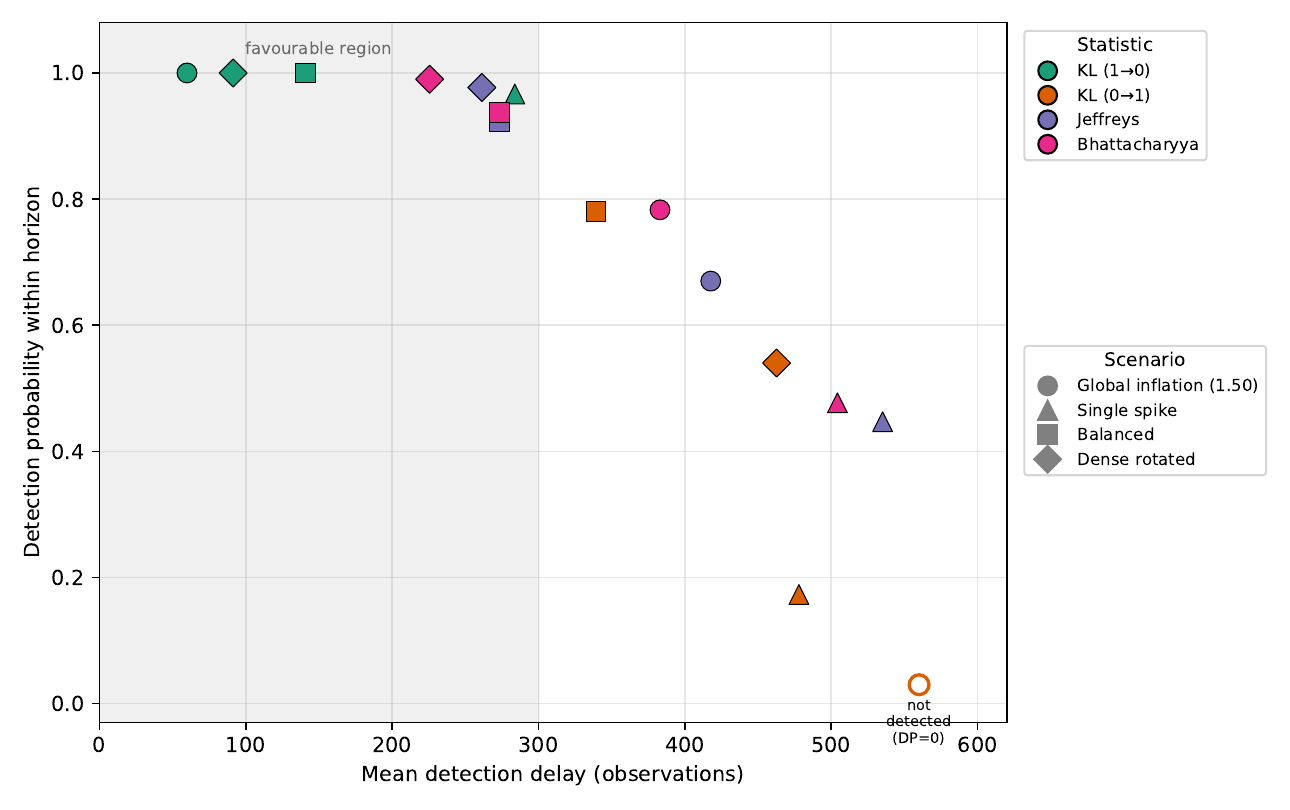}
    \caption{For $k=50$, the power-delay trade-off for four
             spectral deviation statistics calibrated using bootstrap
             is plotted across covariance change scenarios. The points in the upper-left corner
             indicate a high detection probability combined with a short average delay.
             All four series were calculated from the same simulation study
             that forms the basis of Table~\ref{tab:simulation-summary}.}
    \label{fig:simulation-power-delay}
\end{figure}

\end{document}